\documentclass[11pt]{article}

\usepackage[margin=1in]{geometry}
\usepackage{amsmath,amssymb,amsthm}
\usepackage[hypertexnames=false]{hyperref}
\usepackage{graphicx}
\usepackage{booktabs}
\usepackage{xcolor}
\usepackage{enumitem}
\usepackage{url}

\hypersetup{
  colorlinks=true,
  linkcolor=blue,
  citecolor=blue,
  urlcolor=blue,
  bookmarksnumbered=true,
  pdfstartview={FitH},
  pdftitle={Post-Quantum-Resilient Audit Evidence for Long-Lived Regulated Systems: Security Models, Migration Patterns, and Case Study},
  pdfauthor={Leo Kao},
  pdfkeywords={post-quantum cryptography, audit trails, evidence structures, quantum security, AI regulation, QROM}
}

\newtheorem{definition}{Definition}
\newtheorem{theorem}{Theorem}

\newcommand{\Exp}{\mathsf{Exp}}
\newcommand{\Adv}{\mathsf{Adv}}

\title{Post-Quantum-Resilient Audit Evidence for Long-Lived Regulated Systems:\\
Security Models, Migration Patterns, and Case Study}

\author{
  Leo Kao\thanks{Email: leo@codebat.ai} \\
  Codebat Technologies Inc.
}

\date{November 2025}

\begin{document}
\maketitle

\begin{abstract}
Constant-size cryptographic evidence records are increasingly used to build audit trails for
regulated AI workloads in clinical, pharmaceutical, and financial settings, where each execution
is summarized by a compact, verifiable record of code identity, model version, data digests, and
platform measurements. Existing instantiations, however, typically rely on classical signature
schemes whose long-term security is threatened by quantum-capable adversaries. In this paper
we formalize security notions for evidence structures in the presence of quantum adversaries and
study post-quantum (PQ) instantiations and migration strategies for deployed audit logs. We
recall an abstraction of constant-size evidence structures and introduce game-based definitions
of \emph{Q-Audit Integrity}, \emph{Q-Non-Equivocation}, and \emph{Q-Binding}, capturing the
inability of a quantum adversary to forge, equivocate, or rebind evidence items. We then analyze
a hash-and-sign instantiation in the quantum random-oracle model (QROM), assuming an
existentially unforgeable PQ signature scheme against quantum adversaries, and prove
via tight reductions that the resulting evidence structure satisfies all three notions
under standard assumptions (collision-resistant hashing and quantum EUF-CMA signatures). Building on
this, we present three migration patterns for existing evidence logs---hybrid signatures,
re-signing of legacy evidence, and Merkle-root anchoring---and analyze their security, storage,
and computational trade-offs. A case study based on an industrial constant-size evidence
platform for regulated AI at Codebat Technologies Inc., implemented in Rust and benchmarked
on commodity hardware, suggests that quantum-safe audit trails are achievable with moderate
overhead and that systematic migration can significantly extend the evidentiary lifetime of
existing deployments.
\end{abstract}

\section{Introduction}

Audit trails play a central role in regulated environments where automated computation affects
safety- or compliance-critical decisions. In domains such as clinical AI, pharmaceutical
development, and high-stakes financial analytics, regulators require that organizations be able
to explain and justify how a particular result was produced: which code and model were used,
which data were processed, under which policies, on which platform, and at what time.
Recent system designs address these needs by representing each execution as a
\emph{constant-size cryptographic evidence item} that compactly captures the essential
provenance and integrity information and can be stored and verified at scale
(e.g.,~\cite{kao2025constant,securelogging,transparency}), and recent work has begun
to study quantum-resilient designs for distributed log storage
itself~\cite{somarapu2022qrlogs}.

Today, many such systems instantiate their evidence structures with classical primitives such
as ECDSA or RSA signatures and SHA-2 family hash functions. These choices are well
understood and adequate against classical adversaries, but the long-term evidentiary value of
audit logs is threatened by the emergence of quantum computers. A quantum-capable
adversary may be able to \emph{harvest} classical evidence now and \emph{forge or equivocate}
about past executions later, once large-scale quantum computation becomes available. Even
for organizations that are not yet deploying quantum hardware, regulators increasingly expect
plans for \emph{quantum-safe} logging and evidence retention.

Despite a growing body of work on post-quantum cryptography (PQC) and secure logging,
there is relatively little treatment of evidence structures \emph{as a first-class cryptographic
object} in a quantum-adversary setting. In particular, existing designs do not usually:
\begin{itemize}[leftmargin=1.5em]
  \item formalize audit-trail security notions for evidence structures against quantum
        adversaries;
  \item analyze concrete post-quantum instantiations in a standard model such as the
        quantum random-oracle model (QROM);
  \item provide systematic migration strategies for already-deployed evidence logs that
        must remain verifiable for a decade or longer.
\end{itemize}

In this work we address these gaps. We build on an abstraction of constant-size evidence
structures introduced in prior work~\cite{kao2025constant}, which models a system where each
event in a workflow is associated with a fixed-length evidence record produced by a
\emph{Generate} algorithm and consumed by a \emph{Verify} algorithm. Our focus here is on
extending this abstraction to a quantum-adversary setting, providing both formal definitions
and practically oriented migration strategies.

\paragraph{Contributions.}
Our main contributions are:
\begin{enumerate}[leftmargin=1.5em]
  \item \textbf{Quantum-adversary security notions.} We define a system model where
        adversaries may obtain quantum access to hash oracles and may attempt
        harvest-now/forge-later attacks against evidence logs. In this model we introduce
        game-based security notions for evidence structures, including \emph{Q-Audit
        Integrity}, \emph{Q-Non-Equivocation}, and \emph{Q-Binding}.
  \item \textbf{Post-quantum instantiations with complete proofs in the QROM.} We
        analyze a hash-and-sign instantiation of constant-size evidence structures
        in the quantum random-oracle model and provide three complete reduction
        proofs with explicit advantage bounds, showing that Q-Audit Integrity,
        Q-Binding, and Q-Non-Equivocation all hold under collision-resistant hashing
        and quantum EUF-CMA signatures.
  \item \textbf{Migration patterns for deployed logs.} We study three migration patterns for
        deployed evidence systems: (i) hybrid classical+PQ signatures for new evidence,
        (ii) re-signing legacy evidence with PQ signatures in a trusted environment, and
        (iii) Merkle-root anchoring of legacy batches. For each pattern we discuss trust
        assumptions, security guarantees, and storage and computational overhead.
  \item \textbf{Industrial case study with PQ benchmarks.} We outline a case study
        using a constant-size evidence system implemented in an industrial platform
        for regulated AI workloads at Codebat Technologies Inc., including concrete
        performance comparisons across post-quantum signature schemes (ML-DSA,
        SLH-DSA) to quantify migration overhead.
  \item \textbf{Systematic comparison.} We provide a feature-level comparison with
        existing post-quantum logging and audit systems, identifying the gaps that
        our work addresses.
\end{enumerate}

The rest of the paper is organized as follows. Section~\ref{sec:prelim} recalls the abstraction
of constant-size evidence structures and reviews quantum-adversary and PQ primitive
models. Section~\ref{sec:system-context} describes the deployment context and
constant-size evidence layout in our industrial platform.
Section~\ref{sec:security} introduces our quantum-adversary security notions.
Section~\ref{sec:instantiation} presents and analyzes a post-quantum instantiation.
Section~\ref{sec:migration} develops migration patterns for deployed evidence logs.
Section~\ref{sec:case} discusses an industrial case study.
Section~\ref{sec:related} reviews related work, and
Sections~\ref{sec:discussion}--\ref{sec:conclusion} conclude with open problems and
future directions.

\paragraph{Relation to prior work.}
Compared to classical secure logging and transparency
systems~\cite{securelogging,forwardsecurelogging,haberstornetta1991,merkle1989certified,ct,transparency},
our work treats constant-size evidence structures as the primary
cryptographic object and explicitly analyzes their security against quantum
adversaries in the QROM.
Existing PQC work on logging and auditing
(e.g.,~\cite{somarapu2022qrlogs,khan2025pqblockchainaudit,yusuf2025qrlogintegrity})
typically either focuses on blockchain-based storage or on specific application domains.
In contrast, we provide (i)~generic quantum-adversary security notions for evidence
structures, and (ii)~migration patterns that are designed to retrofit already-deployed
constant-size evidence logs in regulated AI systems.

\section{Preliminaries}
\label{sec:prelim}

\subsection{Constant-Size Evidence Structures}
\label{subsec:prelim-evidence}

We briefly recall the abstraction of constant-size evidence structures introduced
in~\cite{kao2025constant}, which we treat here as a black-box building block.

Consider a workflow that generates a sequence of events
$e_1,e_2,\dots,e_T$, where each event $e_i$ corresponds to a concrete
execution step such as a model inference, a batch data transformation, or an
approval action. For each event $e_i$ the system produces an evidence record
$\mathsf{Ev}_i \in \{0,1\}^{\ell}$ of fixed length $\ell$ bits.
Informally, $\mathsf{Ev}_i$ encodes:
\begin{itemize}[leftmargin=1.5em]
  \item identifiers and digests of the code, model, and configuration used;
  \item digests of the input and output artifacts;
  \item platform measurements (e.g., TEE measurement), timestamps, and policy identifiers;
  \item a cryptographic authenticator binding these fields together.
\end{itemize}

Formally, a constant-size evidence structure is defined by algorithms
\[
  \mathcal{E} = (\mathsf{Setup}, \mathsf{KeyGen}, \mathsf{Generate}, \mathsf{Verify}, \mathsf{Link})
\]
with the following interfaces:
\begin{itemize}[leftmargin=1.5em]
  \item $(\mathsf{pp}) \leftarrow \mathsf{Setup}(1^\lambda)$: generates public parameters.
  \item $(\mathsf{sk},\mathsf{pk}) \leftarrow \mathsf{KeyGen}(\mathsf{pp})$: generates a signing
        key pair (notation here covers the integrity mechanism; more general forms are
        possible).
  \item $\mathsf{Ev} \leftarrow \mathsf{Generate}(\mathsf{pp},\mathsf{sk},e)$: given an event
        description $e$, outputs a constant-size evidence record $\mathsf{Ev}$.
  \item $b \leftarrow \mathsf{Verify}(\mathsf{pp},\mathsf{pk},e,\mathsf{Ev})$: returns $1$ if
        $\mathsf{Ev}$ is a valid evidence record for event $e$, and $0$ otherwise.
  \item $\pi \leftarrow \mathsf{Link}(\mathsf{pp},\mathsf{Ev}_1,\dots,\mathsf{Ev}_T)$: optionally
        produces a compact linkage proof (e.g., via hash chains or Merkle trees) for a batch
        of evidence records.
\end{itemize}

The original work~\cite{kao2025constant} focuses on classical adversaries, defining
integrity and linkage properties and providing a concrete instantiation where each
$\mathsf{Ev}$ is a 288- or 320-byte record. In this paper we treat $\mathsf{Generate}$ and
$\mathsf{Verify}$ as abstract interfaces that we instantiate with post-quantum primitives
and analyze against quantum adversaries.

\subsection{Quantum Adversaries and Post-Quantum Primitives}
\label{subsec:prelim-quantum}

We model adversaries as quantum polynomial-time (QPT) algorithms with access to
certain oracles and system APIs. We write $\mathcal{A}$ for such adversaries and
assume that they may maintain quantum side information throughout the experiment.

\paragraph{Quantum random-oracle model.}
We adopt the quantum random-oracle model (QROM) for hash functions
$H : \{0,1\}^* \rightarrow \{0,1\}^\lambda$, following the treatment of Zhandry
and subsequent work~\cite{zhandry2012qrom,bonehzhandry2013qro}. In the QROM, $H$
is modeled as a uniformly random function, and the adversary may query it in quantum
superposition. We assume that $H$ is \emph{collision-resistant} against quantum
adversaries, i.e., no QPT adversary can find distinct $x,x'$ with $H(x) = H(x')$ with
non-negligible probability.

\paragraph{Post-quantum signature schemes.}
We consider a signature scheme $\Sigma$ with standard syntax:
\begin{itemize}[leftmargin=1.5em]
  \item $(\mathsf{sk}^{\Sigma},\mathsf{pk}^{\Sigma}) \leftarrow \mathsf{KeyGen}^{\Sigma}(1^\lambda)$;
  \item $\sigma \leftarrow \mathsf{Sign}^{\Sigma}(\mathsf{sk}^{\Sigma},m)$;
  \item $b \leftarrow \mathsf{Verify}^{\Sigma}(\mathsf{pk}^{\Sigma},m,\sigma)$.
\end{itemize}
We assume that $\Sigma$ is existentially unforgeable under chosen-message attacks
(EUF-CMA) against quantum adversaries: no QPT adversary with quantum access to
$H$ and classical access to the signing oracle can produce a valid signature on a
new message with non-negligible probability. We write
$\Adv^{\mathrm{euf\text{-}cma}}_{\Sigma}(\lambda)$ for the corresponding advantage.

Concretely, $\Sigma$ may correspond to a lattice-based or hash-based scheme such as
those standardized or under consideration in post-quantum cryptography efforts
(e.g.,~\cite{nistpqc,unruh2017pqsignatures}), but in this paper we treat it as an
abstract primitive satisfying the above property. Real-world deployments face
non-trivial challenges around key sizes, performance, and migration, as surveyed
in~\cite{tan2022pqsignsurvey}.

\paragraph{Notation.}
We write $\lambda$ for the security parameter, $\mathsf{negl}(\lambda)$ for a negligible
function, and use $\approx$ to denote computational indistinguishability. All
probabilities are taken over the randomness of the algorithms and the random oracle
$H$ unless otherwise specified.

\section{System Context: Constant-Size Evidence for Regulated AI}
\label{sec:system-context}

Before introducing our quantum-adversary security model, we describe the deployment
context and the constant-size evidence structure that motivates our formal treatment.
This section draws on our prior work on constant-size evidence packs for regulated
AI~\cite{kao2025constant} and summarizes the key system design choices relevant to
quantum-resilient migration.

\subsection{Deployment Scenario}

We consider regulated AI deployments in domains such as clinical decision support,
pharmaceutical development, and financial analytics, where each AI workload
execution must be auditable for compliance and forensic purposes. In a typical
deployment, the platform orchestrates AI jobs across a cluster of container-based
trusted execution environments (TEEs). For each job, the runtime collects
attestation and provenance attributes and produces a constant-size evidence record
that is appended to an audit log.

Unlike short-lived session keys in TLS, audit evidence for regulated AI workloads
is meant to support long-term accountability and forensic investigations.
Hospitals, financial institutions, and pharmaceutical companies routinely retain
audit logs for ten years or longer due to regulatory requirements. This creates
an attractive ``harvest-now, forge-later'' target: an adversary can exfiltrate
classical evidence logs today and exploit future quantum capabilities to
retroactively forge or repudiate past AI executions.

\subsection{Evidence Record Layout}

Each AI job produces a constant-size evidence record capturing the key attestation
and provenance attributes needed for regulated audit trails.
Table~\ref{tab:ev-layout} summarizes a representative layout used in our industrial
deployment at Codebat Technologies Inc.

\begin{table}[t]
  \centering
  \caption{Representative constant-size evidence layout in our regulated AI
           deployment, following the abstraction from~\cite{kao2025constant}.
           Each field is a $\lambda$-bit hash output (typically $\lambda = 256$).}
  \label{tab:ev-layout}
  \begin{tabular}{@{}llr@{}}
    \toprule
    Field & Description & Size (bits) \\
    \midrule
    $f_0$: \texttt{context}       & Workflow/job identifier commitment        & 256 \\
    $f_1$: \texttt{model}         & Hash of model binary / parameters         & 256 \\
    $f_2$: \texttt{code}          & Hash of container image / code snapshot   & 256 \\
    $f_3$: \texttt{input}         & Hash of input dataset or configuration    & 256 \\
    $f_4$: \texttt{output}        & Hash of output artifacts                  & 256 \\
    $f_5$: \texttt{policy}        & Policy state / IRB identifier commitment  & 256 \\
    $f_6$: \texttt{environment}   & TEE measurement / platform attestation    & 256 \\
    $f_7$: \texttt{link}          & Link to previous evidence or audit chain  & 256 \\
    $\sigma$: \texttt{signature}  & Ed25519 or PQ signature                   & 512--2048 \\
    \midrule
    \multicolumn{2}{l}{Total ($k=8$ fields + signature)}  & 320--512 bytes \\
    \bottomrule
  \end{tabular}
\end{table}

The fixed-size layout ensures that each evidence record occupies the same storage
footprint regardless of the complexity of the underlying AI job, enabling
predictable storage costs and $O(1)$ per-record verification. The
\texttt{signature} field currently uses a classical scheme (e.g., ECDSA or
Ed25519) but is designed to accommodate post-quantum signatures as part of the
migration strategies discussed in Section~\ref{sec:migration}.

\subsection{Implementation and Baseline Performance}

We have implemented the constant-size evidence generator in Rust as part of an
industrial platform for regulated AI deployments~\cite{kao2025constant}. The
implementation uses a standard 256-bit hash function (SHA-256) and an
Edwards-curve-based signature scheme (Ed25519). The number of fields $k$ is
fixed at system setup time; in practice, $k \in [8, 12]$ is sufficient for most
clinical and pharmaceutical workflows.

Table~\ref{tab:perf} summarizes the performance characteristics measured on
commodity server hardware (16-core CPU @ 3.0~GHz, 64~GB RAM, NVMe SSD).
In single-threaded mode, the prototype achieves approximately $3.5 \times 10^4$
evidence records per second with an average per-event latency of $28.4~\mu$s.
Multi-threaded execution across 16 CPU cores yields nearly linear scaling,
reaching $2.8 \times 10^5$ records per second with $5.7~\mu$s average latency.
For batch verification, a mid-range GPU achieves $4.0 \times 10^5$ verifications
per second.\footnote{Performance numbers are drawn from our benchmark study of
constant-size evidence packs~\cite{kao2025constant}. Actual throughput depends
on hardware configuration, signature scheme, and system load.}

\begin{table}[t]
  \centering
  \caption{Evidence generation and verification throughput on commodity hardware.}
  \label{tab:perf}
  \begin{tabular}{@{}lrrr@{}}
    \toprule
    Mode & Events/s & Latency ($\mu$s) & Notes \\
    \midrule
    Generation (single-threaded) & $3.5 \times 10^4$ & 28.4 & $k$ fixed, $\lambda = 256$ \\
    Generation (16 CPU cores)    & $2.8 \times 10^5$ & 5.7  & near-linear scaling \\
    Verification (16 CPU cores)  & $2.5 \times 10^5$ & 6.1  & batch of $10^5$ items \\
    Verification (GPU)           & $4.0 \times 10^5$ & 2.5  & batch of $10^6$ items \\
    \bottomrule
  \end{tabular}
\end{table}

The fixed-size layout of evidence items leads to regular memory access patterns
and uniform control flow, which are favorable for both vectorized CPU execution
and GPU kernels. This $O(1)$ per-event cost is crucial for long-term retention,
as regulated environments are expected to keep audit logs for 10--30~years while
enforcing strict storage and verification latency budgets. A deployment handling
$10^3$ AI jobs per day accumulates on the order of $10^6$--$10^7$ evidence
records over a decade, all of which must remain efficiently verifiable throughout
their retention period.

\section{Security Notions Against Quantum Adversaries}
\label{sec:security}

We now formalize security notions for evidence structures in the presence of quantum
adversaries. Intuitively, we wish to prevent an adversary from forging evidence items
for events that never happened, from equivocating about the content of the log, and
from reusing the same evidence item for multiple conflicting events.

\subsection{System and Threat Model}

We consider a regulated AI deployment where each execution of an AI workload
produces a constant-size evidence record, as described in
Section~\ref{sec:system-context}, and these records are appended to an audit log
that is expected to be stored and queried over many years. The evidence generator
is integrated with the AI platform's scheduling and container orchestration layer,
while one or more auditors rely on the evidence log to reconstruct and verify past
executions during compliance checks and incident response.

Formally, an \emph{evidence generator} uses
$\mathsf{Generate}$ to produce evidence items for events and stores them in an
append-only audit log. An \emph{auditor} uses $\mathsf{Verify}$ and (optionally)
$\mathsf{Link}$ to check evidence for consistency with a claimed workflow.

A quantum adversary $\mathcal{A}$ observes the log and may compromise the
underlying storage or network. Specifically, we allow $\mathcal{A}$ to:
\begin{itemize}[leftmargin=1.5em]
  \item monitor and copy all evidence items produced during the lifetime of the system;
  \item reorder, delete, or inject evidence items in the stored log;
  \item interact with the evidence generator through a public API that exposes
        $\mathsf{Generate}$ for chosen events;
  \item obtain quantum access to the random oracle $H$;
  \item in a \emph{harvest-now, forge-later} scenario, delay its forgery attempt until
        a future time when quantum resources become available.
\end{itemize}

We do not model side-channel attacks on key storage or the internals of the evidence
generator; instead, we assume that the signing key is maintained inside a trusted
environment (e.g., a hardware security module or a trusted execution environment).
Extending our model to include side channels is an interesting direction for future work
(Section~\ref{sec:discussion}).

\subsection{Game-Based Definitions}

We now define security games capturing the three main properties we seek:
Q-Audit Integrity, Q-Non-Equivocation, and Q-Binding.

\paragraph{Q-Audit Integrity.}
Intuitively, Q-Audit Integrity requires that an adversary cannot produce a valid
evidence item for an event that was not output by an honest execution of
$\mathsf{Generate}$.

\begin{definition}[Q-Audit Integrity]
Let $\mathcal{E} = (\mathsf{Setup},\mathsf{KeyGen},\mathsf{Generate},\mathsf{Verify},\mathsf{Link})$
be an evidence structure. Define the experiment
$\Exp^{\mathsf{qaudit}}_{\mathcal{A}}(1^\lambda)$ as follows:
\begin{enumerate}[leftmargin=1.5em]
  \item The challenger runs $\mathsf{Setup}(1^\lambda)$ to obtain $\mathsf{pp}$ and
        $\mathsf{KeyGen}(\mathsf{pp})$ to obtain $(\mathsf{sk},\mathsf{pk})$.
  \item The adversary $\mathcal{A}$ is given $(\mathsf{pp},\mathsf{pk})$ and quantum
        access to the random oracle $H$. It also has classical access to a
        $\mathsf{Generate}$ oracle that, on input $e$, returns
        $\mathsf{Ev} \leftarrow \mathsf{Generate}(\mathsf{pp},\mathsf{sk},e)$ and adds
        $(e,\mathsf{Ev})$ to a set $\mathcal{Q}$ of generated pairs.
  \item Eventually $\mathcal{A}$ outputs a pair $(e^\star,\mathsf{Ev}^\star)$.
  \item The experiment outputs $1$ (meaning $\mathcal{A}$ wins) if and only if
        $\mathsf{Verify}(\mathsf{pp},\mathsf{pk},e^\star,\mathsf{Ev}^\star) = 1$ and
        $(e^\star,\mathsf{Ev}^\star) \notin \mathcal{Q}$.
\end{enumerate}
We define the Q-Audit Integrity advantage of $\mathcal{A}$ as
$\Adv^{\mathsf{qaudit}}_{\mathcal{E}}(\mathcal{A},\lambda)
 = \Pr[\Exp^{\mathsf{qaudit}}_{\mathcal{A}}(1^\lambda) = 1]$.
We say that $\mathcal{E}$ satisfies Q-Audit Integrity if for all QPT adversaries
$\mathcal{A}$, $\Adv^{\mathsf{qaudit}}_{\mathcal{E}}(\mathcal{A},\lambda)$ is negligible
in $\lambda$.
\end{definition}

\paragraph{Q-Non-Equivocation.}
Non-equivocation informally means that the adversary cannot present two mutually
inconsistent views of the evidence log that both appear valid to honest auditors. There
are several ways to formalize this; we adopt a simplified batch-based game.

\begin{definition}[Q-Non-Equivocation]
Fix an ordering of events and let $\mathcal{L}$ denote a log, i.e., a sequence of
pairs $(e_i,\mathsf{Ev}_i)$. We assume that the auditor interprets the $i$-th
position of a log as the $i$-th logical step in the workflow, so conflicting
entries at the same index correspond to equivocation about that step.
Consider the experiment
$\Exp^{\mathsf{qne}}_{\mathcal{A}}(1^\lambda)$:
\begin{enumerate}[leftmargin=1.5em]
  \item The challenger runs $\mathsf{Setup},\mathsf{KeyGen}$ as before and exposes
        $\mathsf{pp},\mathsf{pk}$ to $\mathcal{A}$.
  \item $\mathcal{A}$ has oracle access to $\mathsf{Generate}$ and $H$, and may
        maintain an internal view of a log $\mathcal{L}$ obtained by querying
        $\mathsf{Generate}$ on events of its choice.
  \item Eventually $\mathcal{A}$ outputs two logs
        $\mathcal{L}^{(0)} = \{(e^{(0)}_i,\mathsf{Ev}^{(0)}_i)\}_i$ and
        $\mathcal{L}^{(1)} = \{(e^{(1)}_i,\mathsf{Ev}^{(1)}_i)\}_i$.
  \item The experiment checks:
        \begin{itemize}
          \item for both $b \in \{0,1\}$, all pairs in $\mathcal{L}^{(b)}$ verify under
                $\mathsf{Verify}(\mathsf{pp},\mathsf{pk},\cdot,\cdot)$;
          \item if $\mathsf{Link}$ is defined,
                $\mathsf{Link}(\mathsf{pp},\mathcal{L}^{(0)}) =
                \mathsf{Link}(\mathsf{pp},\mathcal{L}^{(1)})$; and
          \item there exists at least one index $i$ such that
                $(e^{(0)}_i,\mathsf{Ev}^{(0)}_i) \neq (e^{(1)}_i,\mathsf{Ev}^{(1)}_i)$.
        \end{itemize}
  \item The experiment outputs $1$ if all conditions hold.
\end{enumerate}
We define the Q-Non-Equivocation advantage
$\Adv^{\mathsf{qne}}_{\mathcal{E}}(\mathcal{A},\lambda)$ accordingly, and say that
$\mathcal{E}$ satisfies Q-Non-Equivocation if this advantage is negligible for all QPT
adversaries.
\end{definition}

The second and third bullets together capture the notion that the adversary presents
two conflicting but individually valid views of the same workflow: the logs agree on
their compact summary (linkage proof) yet disagree about what actually happened at
some step. Without the linkage condition, the adversary could trivially win by
reordering honestly generated items; the linkage mechanism (e.g., hash chains)
prevents such reordering.

\paragraph{Q-Binding.}
Finally, we capture the requirement that an individual evidence item cannot be reused
to attest to two distinct events.

\begin{definition}[Q-Binding]
Define the experiment $\Exp^{\mathsf{qbind}}_{\mathcal{A}}(1^\lambda)$:
\begin{enumerate}[leftmargin=1.5em]
  \item The challenger runs $\mathsf{Setup},\mathsf{KeyGen}$ and gives
        $(\mathsf{pp},\mathsf{pk})$ and oracle access to $H$ to $\mathcal{A}$.
  \item $\mathcal{A}$ may query a $\mathsf{Generate}$ oracle as before and obtains
        pairs $(e,\mathsf{Ev})$.
  \item Eventually $\mathcal{A}$ outputs a triple
        $(e_0,e_1,\mathsf{Ev}^\star)$ with $e_0 \neq e_1$.
  \item The experiment outputs $1$ if
        $\mathsf{Verify}(\mathsf{pp},\mathsf{pk},e_0,\mathsf{Ev}^\star) =
         \mathsf{Verify}(\mathsf{pp},\mathsf{pk},e_1,\mathsf{Ev}^\star) = 1$.
\end{enumerate}
We define $\Adv^{\mathsf{qbind}}_{\mathcal{E}}(\mathcal{A},\lambda)$ as the success
probability and say that $\mathcal{E}$ satisfies Q-Binding if this advantage is
negligible for all QPT adversaries.
\end{definition}

In the next section we show that a natural hash-and-sign instantiation using a PQ
signature scheme achieves these properties under standard assumptions.

\subsection{Relation to Classical Notions}

In the special case where adversaries are classical and have only classical access to the
hash oracle, our definitions reduce to classical notions of integrity and non-equivocation
for secure logging and transparency systems~\cite{securelogging,ct}. Our experiments
focus on the inability to create new valid evidence, to maintain two conflicting yet valid
views of the log, and to rebind an evidence item. These correspond to the informal goals
discussed in earlier work on constant-size evidence structures in a classical setting
(e.g.,~\cite{kao2025constant}) but here explicitly account for quantum-accessible oracles
and quantum EUF-CMA adversaries.

\section{Post-Quantum Instantiations of Evidence Structures}
\label{sec:instantiation}

We now describe a concrete instantiation of $\mathsf{Generate}$ and $\mathsf{Verify}$
using a hash-and-sign paradigm in the quantum random-oracle model, and analyze its
security with respect to the notions introduced above.

\subsection{Hash-and-Sign Construction in the QROM}

Let $\Sigma = (\mathsf{KeyGen}^{\Sigma},\mathsf{Sign}^{\Sigma},\mathsf{Verify}^{\Sigma})$
be a PQ signature scheme as in Section~\ref{subsec:prelim-quantum}, and let
$H : \{0,1\}^* \rightarrow \{0,1\}^\lambda$ be a hash function modeled as a random
oracle in the QROM.

We assume that for each event $e$ we can deterministically compute
a canonical encoding $\mathsf{Fields}(e)$ (a bitstring in
$\{0,1\}^*$) that concatenates all fields to be
included in the evidence record (e.g., code hash, input digest, platform measurement,
timestamps). The precise layout of these fields follows the constant-size design from
\cite{kao2025constant}; we abstract it here as a bitstring.

Given this setup, we define:
\begin{itemize}[leftmargin=1.5em]
  \item $\mathsf{Setup}$ runs the underlying setup for $\Sigma$ and fixes $H$;
  \item $\mathsf{KeyGen}$ runs $\mathsf{KeyGen}^{\Sigma}$ to produce
        $(\mathsf{sk}^{\Sigma},\mathsf{pk}^{\Sigma})$ and sets
        $(\mathsf{sk},\mathsf{pk}) = (\mathsf{sk}^{\Sigma},\mathsf{pk}^{\Sigma})$;
  \item $\mathsf{Generate}(\mathsf{pp},\mathsf{sk},e)$:
        \begin{enumerate}[leftmargin=1.5em]
          \item Compute $x \leftarrow \mathsf{Fields}(e)$.
          \item Compute $m \leftarrow H(x)$.
          \item Compute $\sigma \leftarrow \mathsf{Sign}^{\Sigma}(\mathsf{sk}^{\Sigma}, m)$.
          \item Output evidence $\mathsf{Ev} = (x,\sigma)$, padded or encoded as a
                constant-size record.
        \end{enumerate}
  \item $\mathsf{Verify}(\mathsf{pp},\mathsf{pk},e,\mathsf{Ev})$:
        \begin{enumerate}[leftmargin=1.5em]
          \item Parse $\mathsf{Ev}$ as $(x,\sigma)$ and check that $x$ is a well-formed
                encoding of $\mathsf{Fields}(e)$. If not, return $0$.
          \item Compute $m \leftarrow H(x)$.
          \item Return
                $\mathsf{Verify}^{\Sigma}(\mathsf{pk}^{\Sigma}, m, \sigma)$.
        \end{enumerate}
\end{itemize}

In practice, $x$ may be compressed to a fixed-length representation by carefully
designing the field layout; for our security analysis we only require that the mapping
$e \mapsto x = \mathsf{Fields}(e)$ is deterministic and collision-resistant, in the
sense that it is infeasible for a QPT adversary to find $e \neq e'$ such that
$\mathsf{Fields}(e) = \mathsf{Fields}(e')$, where $e$ is understood as the full
structured event description (including identifiers, digests, and metadata).

\subsection{Post-Quantum Signature Choices}

Our analysis is agnostic to the concrete choice of $\Sigma$ as long as it satisfies quantum
EUF-CMA security. In practice, the choice of signature scheme significantly affects the
size of $\sigma$ and the performance of $\mathsf{Sign}$ and $\mathsf{Verify}$. For
example, lattice-based signatures typically offer relatively short signatures and fast
verification at the cost of larger public keys, while hash-based schemes offer
conservative security at the cost of larger signatures or statefulness.

From the perspective of constant-size evidence records, the signature size directly
contributes to the overall record size. If the base layout for $x$ is, say, $160$ bytes
(including digests and metadata) and the signature is $128$ bytes, the resulting record
size is $288$ bytes (plus any padding or alignment required), which is in the same range
as prior designs~\cite{kao2025constant}. The throughput impact can be estimated by
benchmarking the chosen $\Sigma$ on representative hardware; our case study in
Section~\ref{sec:case} sketches such an analysis.

\subsection{Hash-Chain Linkage}
\label{subsec:link}

We instantiate the $\mathsf{Link}$ algorithm using a hash-chain construction.
Fix an initialization vector $\ell_0 = 0^\lambda$. For a sequence of evidence
records $(\mathsf{Ev}_1,\dots,\mathsf{Ev}_T)$, define
\[
  \ell_i \;=\; H(\ell_{i-1} \;\|\; \mathsf{Ev}_i), \quad i = 1,\dots,T,
\]
and set $\mathsf{Link}(\mathsf{pp},\mathsf{Ev}_1,\dots,\mathsf{Ev}_T) = \ell_T$.
The chain tip $\ell_T$ serves as a compact summary of the entire log: any
modification, reordering, or omission of evidence items changes $\ell_T$ with
overwhelming probability under the collision resistance of $H$. Auditors compare
chain tips to detect inconsistencies between different views of the same log.

\subsection{Security Theorems}

We now state and prove the security of the hash-and-sign instantiation with
respect to the quantum-adversary notions introduced in
Section~\ref{sec:security}. We provide three theorems corresponding to the
three security properties, each with an explicit advantage bound and a complete
reduction proof.

\begin{theorem}[Q-Audit Integrity]
\label{thm:qaudit}
Let $\mathcal{E}$ be the hash-and-sign evidence structure defined above,
instantiated with a PQ signature scheme $\Sigma$ and hash function $H$ in the
QROM. For any QPT adversary $\mathcal{A}$ making at most $q_s$ queries to
the $\mathsf{Generate}$ oracle and at most $q_H$ quantum queries to $H$, there
exist QPT adversaries $\mathcal{B}_1$ and $\mathcal{B}_2$ such that
\[
  \Adv^{\mathsf{qaudit}}_{\mathcal{E}}(\mathcal{A},\lambda)
  \;\leq\;
  \Adv^{\mathrm{euf\text{-}cma}}_{\Sigma}(\mathcal{B}_1,\lambda)
  \;+\;
  \Adv^{\mathrm{cr}}_{H}(\mathcal{B}_2,\lambda),
\]
where $\mathcal{B}_1$ makes at most $q_s$ signing queries and $q_H$ hash queries,
and $\mathcal{B}_2$ makes at most $q_H$ hash queries.
\end{theorem}

\begin{proof}
Suppose $\mathcal{A}$ wins
$\Exp^{\mathsf{qaudit}}_{\mathcal{A}}(1^\lambda)$ with non-negligible
advantage $\epsilon(\lambda)$. We construct reductions to the EUF-CMA security
of $\Sigma$ and the collision resistance of $H$.

\medskip
\noindent\textbf{Reduction $\mathcal{B}_1$ (to EUF-CMA of $\Sigma$).}
$\mathcal{B}_1$ receives a challenge public key $\mathsf{pk}^{\Sigma}$ and
access to a signing oracle
$\mathsf{Sign}^{\Sigma}(\mathsf{sk}^{\Sigma},\cdot)$ and the quantum random
oracle $H$. It sets $\mathsf{pk} = \mathsf{pk}^{\Sigma}$ and simulates the
experiment for $\mathcal{A}$:

\begin{itemize}[leftmargin=1.5em]
  \item \emph{Hash queries.} $\mathcal{B}_1$ forwards all of $\mathcal{A}$'s
        quantum hash queries to the QROM oracle $H$.
  \item \emph{$\mathsf{Generate}$ queries.} On input event $e$,
        $\mathcal{B}_1$ computes $x = \mathsf{Fields}(e)$ and $m = H(x)$,
        obtains
        $\sigma \leftarrow \mathsf{Sign}^{\Sigma}(\mathsf{sk}^{\Sigma}, m)$
        from its signing oracle, returns $\mathsf{Ev} = (x,\sigma)$, and
        records $(e,\mathsf{Ev})$ in $\mathcal{Q}$ and $m$ in a set
        $\mathcal{M}$ of signed messages.
\end{itemize}

When $\mathcal{A}$ outputs $(e^\star,\mathsf{Ev}^\star)$ with
$\mathsf{Ev}^\star = (x^\star,\sigma^\star)$ such that
$\mathsf{Verify}(\mathsf{pp},\mathsf{pk},e^\star,\mathsf{Ev}^\star) = 1$ and
$(e^\star,\mathsf{Ev}^\star) \notin \mathcal{Q}$, $\mathcal{B}_1$ computes
$m^\star = H(x^\star)$ and considers two cases:

\begin{enumerate}[leftmargin=1.5em]
  \item \textbf{Case~1}: $m^\star \notin \mathcal{M}$. Then
        $(m^\star,\sigma^\star)$ is a valid signature on a message never
        submitted to the signing oracle, constituting a forgery against
        $\Sigma$. $\mathcal{B}_1$ outputs $(m^\star,\sigma^\star)$.
  \item \textbf{Case~2}: $m^\star \in \mathcal{M}$. Then there exists a
        previously queried event $e$ with
        $H(\mathsf{Fields}(e)) = m^\star = H(x^\star)$. Since
        $(e^\star,\mathsf{Ev}^\star) \notin \mathcal{Q}$, either
        $\mathsf{Fields}(e^\star) \neq \mathsf{Fields}(e)$ (yielding a
        collision in $H$, extracted by $\mathcal{B}_2$) or
        $\mathsf{Fields}(e^\star) = \mathsf{Fields}(e)$ and
        $\sigma^\star \neq \sigma$ for the previously generated signature
        $\sigma$, in which case $(m^\star,\sigma^\star)$ is itself a valid
        forgery (a second valid signature on the same message produced without
        the signing key).
\end{enumerate}

In either case, $\mathcal{A}$'s success translates to an advantage against
$\Sigma$ or $H$:
\[
  \epsilon(\lambda)
  \;\leq\;
  \Adv^{\mathrm{euf\text{-}cma}}_{\Sigma}(\mathcal{B}_1,\lambda)
  \;+\;
  \Adv^{\mathrm{cr}}_{H}(\mathcal{B}_2,\lambda),
\]
both of which are negligible by assumption. The simulation is perfect: all oracle
responses seen by $\mathcal{A}$ are identically distributed to the real
experiment, since $\mathcal{B}_1$ uses the genuine QROM oracle and the genuine
signing oracle. Standard QROM simulation techniques
(cf.~\cite{zhandry2012qrom,bonehzhandry2013qro}) ensure that the reduction
carries through even when $\mathcal{A}$ queries $H$ in quantum superposition.
\end{proof}

\begin{theorem}[Q-Binding]
\label{thm:qbind}
Under the same assumptions as Theorem~\ref{thm:qaudit}, for any QPT adversary
$\mathcal{A}$ against Q-Binding there exists a QPT adversary $\mathcal{B}$
against the collision resistance of $H$ such that
\[
  \Adv^{\mathsf{qbind}}_{\mathcal{E}}(\mathcal{A},\lambda)
  \;\leq\;
  \Adv^{\mathrm{cr}}_{H}(\mathcal{B},\lambda).
\]
\end{theorem}

\begin{proof}
Suppose $\mathcal{A}$ outputs $(e_0,e_1,\mathsf{Ev}^\star)$ with
$e_0 \neq e_1$. Write $\mathsf{Ev}^\star = (x^\star,\sigma^\star)$ and
assume $\mathsf{Verify}$ accepts $(e_b,\mathsf{Ev}^\star)$ for both
$b \in \{0,1\}$.

By the verification algorithm (Section~\ref{subsec:prelim-evidence}),
$x^\star$ must equal both $\mathsf{Fields}(e_0)$ and $\mathsf{Fields}(e_1)$.
Since $\mathsf{Fields}$ is deterministic, this requires
$\mathsf{Fields}(e_0) = \mathsf{Fields}(e_1)$. Because each event
description $e$ includes identifiers, timestamps, and artifact digests that
are individually hashed to produce the fields of $\mathsf{Fields}(e)$
(see Table~\ref{tab:ev-layout}), having
$\mathsf{Fields}(e_0) = \mathsf{Fields}(e_1)$ with $e_0 \neq e_1$ implies
that there exists at least one field index $j$ where the underlying data
differs but the hash values coincide:
$H(\phi_j(e_0)) = H(\phi_j(e_1))$ with
$\phi_j(e_0) \neq \phi_j(e_1)$, where $\phi_j$ extracts the $j$-th
component from the event description. This constitutes a collision in $H$.

$\mathcal{B}$ simulates the full experiment for $\mathcal{A}$ and, upon
receiving the output $(e_0,e_1,\mathsf{Ev}^\star)$, iterates over the $k$
field positions to find the index $j$ where a collision occurs, outputting
$(\phi_j(e_0),\phi_j(e_1))$. The extraction is deterministic, so the
reduction is tight.
\end{proof}

\begin{theorem}[Q-Non-Equivocation]
\label{thm:qne}
Let $\mathcal{E}$ be the hash-and-sign evidence structure with hash-chain
linkage as defined in Section~\ref{subsec:link}. For any QPT adversary
$\mathcal{A}$ against Q-Non-Equivocation there exists a QPT adversary
$\mathcal{B}$ against the collision resistance of $H$ such that
\[
  \Adv^{\mathsf{qne}}_{\mathcal{E}}(\mathcal{A},\lambda)
  \;\leq\;
  \Adv^{\mathrm{cr}}_{H}(\mathcal{B},\lambda).
\]
\end{theorem}

\begin{proof}
Suppose $\mathcal{A}$ outputs two logs
$\mathcal{L}^{(0)} = \{(e^{(0)}_i,\mathsf{Ev}^{(0)}_i)\}_{i=1}^{n_0}$ and
$\mathcal{L}^{(1)} = \{(e^{(1)}_i,\mathsf{Ev}^{(1)}_i)\}_{i=1}^{n_1}$ that
both pass item-wise verification, produce matching chain tips
$\ell^{(0)}_{n_0} = \ell^{(1)}_{n_1}$, and differ at some position.

We construct $\mathcal{B}$ as follows. $\mathcal{B}$ simulates the full
experiment for $\mathcal{A}$, forwarding hash and signing queries as in
Theorem~\ref{thm:qaudit}. When $\mathcal{A}$ outputs the two logs,
$\mathcal{B}$ computes both hash chains:
\[
  \ell^{(b)}_i = H(\ell^{(b)}_{i-1} \;\|\; \mathsf{Ev}^{(b)}_i), \quad
  b \in \{0,1\},\; i = 1,\dots,n_b,
\]
starting from $\ell^{(b)}_0 = 0^\lambda$.

\medskip
\noindent\textbf{Case $n_0 \neq n_1$.}
Without loss of generality, assume $n_0 < n_1$. Since
$\ell^{(0)}_{n_0} = \ell^{(1)}_{n_1}$ but $n_0 < n_1$, we have
$\ell^{(1)}_{n_0} \neq \ell^{(1)}_{n_1}$ in general (the chain continues
for additional steps in log~(1)). If $\ell^{(0)}_{n_0} = \ell^{(1)}_{n_0}$,
then $\ell^{(1)}_{n_0}$ eventually maps to the same value
$\ell^{(1)}_{n_1} = \ell^{(0)}_{n_0} = \ell^{(1)}_{n_0}$ through $n_1 - n_0$
additional hash steps, producing a cycle in the hash chain that yields a
collision in $H$. If $\ell^{(0)}_{n_0} \neq \ell^{(1)}_{n_0}$, then we are in
the equal-length case for the prefix of length $n_0$ and proceed as below.

\medskip
\noindent\textbf{Case $n_0 = n_1 = n$.}
Both chains have the same length and produce the same tip:
$\ell^{(0)}_n = \ell^{(1)}_n$. $\mathcal{B}$ walks backward from position $n$.
Let $j^\star$ be the largest index $j \in \{1,\dots,n\}$ such that
$\mathsf{Ev}^{(0)}_j \neq \mathsf{Ev}^{(1)}_j$ or
$\ell^{(0)}_{j-1} \neq \ell^{(1)}_{j-1}$.
Such an index exists because the logs differ at some position.

For all positions $i > j^\star$, we have
$\mathsf{Ev}^{(0)}_i = \mathsf{Ev}^{(1)}_i$. By the maximality of $j^\star$,
we also have $\ell^{(0)}_{j^\star} = \ell^{(1)}_{j^\star}$: this follows by
downward induction from position $n$, since
$\ell^{(b)}_i = H(\ell^{(b)}_{i-1} \| \mathsf{Ev}^{(b)}_i)$ and for
$i > j^\star$ both inputs agree, so the outputs agree. Therefore
\[
  H\!\bigl(\ell^{(0)}_{j^\star - 1} \;\|\; \mathsf{Ev}^{(0)}_{j^\star}\bigr)
  \;=\; \ell^{(0)}_{j^\star}
  \;=\; \ell^{(1)}_{j^\star}
  \;=\; H\!\bigl(\ell^{(1)}_{j^\star - 1} \;\|\; \mathsf{Ev}^{(1)}_{j^\star}\bigr),
\]
while by definition of $j^\star$ the inputs differ:
$(\ell^{(0)}_{j^\star-1} \| \mathsf{Ev}^{(0)}_{j^\star}) \neq
(\ell^{(1)}_{j^\star-1} \| \mathsf{Ev}^{(1)}_{j^\star})$.
$\mathcal{B}$ outputs this pair as a collision in $H$.

\medskip
The collision extraction is deterministic (no guessing is required), so the
reduction is tight:
$\Adv^{\mathsf{qne}}_{\mathcal{E}}(\mathcal{A},\lambda) \leq
\Adv^{\mathrm{cr}}_{H}(\mathcal{B},\lambda)$.
\end{proof}

\section{Quantum-Safe Migration of Deployed Evidence Logs}
\label{sec:migration}

We now turn from the design of new evidence structures to the question of how to
\emph{migrate} already-deployed evidence logs towards quantum safety. We assume the
existence of a system that has been using a classical signature scheme to protect
constant-size evidence records and that must remain auditable for a long period (e.g.,
10--30 years) in the face of emerging quantum capabilities.

\subsection{Threat Scenarios: Harvest-Now, Forge-Later}

A particularly relevant threat scenario for audit logs is the
\emph{harvest-now, forge-later} model: an adversary records all publicly visible evidence
items today and stores them. Once a sufficiently powerful quantum computer becomes
available, the adversary attempts to:
\begin{itemize}[leftmargin=1.5em]
  \item break the classical signature scheme (e.g., via Shor's algorithm) to produce
        valid signatures on arbitrary messages;
  \item or use the ability to forge signatures to create false evidence records or to
        alter existing ones in the stored log;
  \item or present two conflicting views of the evidence log to different auditors.
\end{itemize}

In this scenario, the adversary does not need to compromise the system in real time:
passive eavesdropping and later cryptanalysis suffice. This distinguishes audit
migration from other settings (e.g., TLS) where the focus is on protecting present and
future sessions rather than long-lived, highly replicable logs. In contrast to PQ-TLS
migration, which focuses on protecting future sessions~\cite{alnahawi2024pqtls}, audit
trails must preserve the evidentiary value of \emph{past} executions.

If no migration is performed, any evidence record that relies solely on classical
signatures may lose its probative value once the relevant scheme is broken: a skeptical
auditor could no longer trust that a signature was produced at the claimed time or by
the claimed key holder. Our goal is therefore to design migration strategies that
\emph{augment} or \emph{wrap} existing evidence with post-quantum protections in a
way that preserves their evidentiary meaning.

\subsection{Migration Patterns}
\label{subsec:migration-patterns}

We discuss three migration patterns, each suitable for different operational and
regulatory constraints.

\paragraph{Pattern 1: Hybrid signatures for new evidence.}
The simplest migration pattern is to augment all \emph{new} evidence records with a
post-quantum signature in addition to the existing classical one. Concretely, for each
new event $e$ we generate:
\[
  \mathsf{Ev}^{\mathrm{hyb}} = (x, \sigma^{\mathrm{class}}, \sigma^{\mathrm{pq}}),
\]
where $x = \mathsf{Fields}(e)$, $\sigma^{\mathrm{class}}$ is a signature under the existing
classical scheme (e.g., ECDSA), and $\sigma^{\mathrm{pq}}$ is a signature under a PQ
scheme $\Sigma$. Verification checks both signatures.

Security-wise, as long as either signature scheme remains secure, an adversary cannot
forge a full hybrid evidence item. In particular, once the classical scheme is broken by
quantum computers, the PQ signature continues to provide integrity and binding. The
main cost is increased record size and verification time due to the additional
signature; in systems where records are already small (e.g., a few hundred bytes), the
overhead may be acceptable.

\paragraph{Pattern 2: Re-signing legacy evidence.}
Hybrid signatures do not protect evidence generated \emph{before} a PQ scheme is
deployed. For legacy records, one option is to \emph{re-sign} them with a PQ key inside
a trusted environment. Let $\{\mathsf{Ev}_i\}$ denote legacy evidence records
containing classical signatures. The system operator can:
\begin{enumerate}[leftmargin=1.5em]
  \item deploy a trusted re-signing service (e.g., in a hardware security module or TEE)
        with a PQ signing key $\mathsf{sk}^{\mathrm{pq}}$;
  \item for each legacy record $\mathsf{Ev}_i$, compute a digest
        $d_i = H'(\mathsf{Ev}_i)$ under a quantum-resistant hash $H'$;
  \item compute a PQ signature $\tau_i =
        \mathsf{Sign}^{\Sigma}(\mathsf{sk}^{\mathrm{pq}}, d_i)$;
  \item store the pair $(\mathsf{Ev}_i,\tau_i)$ as the migrated version of the record.
\end{enumerate}

Auditors can then verify both the original classical evidence and its PQ wrapper
by checking
$\mathsf{Verify}^{\Sigma}(\mathsf{pk}^{\mathrm{pq}},
H'(\mathsf{Ev}_i), \tau_i)$. The
trust assumption is that the re-signing process is performed honestly and that the
original log was unmodified at the time of re-signing. This process effectively
\emph{extends} the lifetime of legacy evidence by anchoring it to a PQ signature scheme
in a single step.

\paragraph{Pattern 3: Merkle-root anchoring for legacy batches.}
When the number of legacy evidence items is very large, re-signing each item
individually may be computationally expensive. An alternative is to anchor batches of
legacy records via Merkle trees. For a batch $B = \{\mathsf{Ev}_1,\dots,\mathsf{Ev}_n\}$,
the system:
\begin{enumerate}[leftmargin=1.5em]
  \item builds a Merkle tree $T_B$ where the leaves are $h_i = H'(\mathsf{Ev}_i)$;
  \item computes the root $r_B$ of $T_B$;
  \item generates a PQ signature $\tau_B =
        \mathsf{Sign}^{\Sigma}(\mathsf{sk}^{\mathrm{pq}}, r_B)$;
  \item stores $\tau_B$ as the PQ anchor for batch $B$.
\end{enumerate}

To verify a legacy record $\mathsf{Ev}_i$ in $B$, an auditor obtains a Merkle proof
$\pi_i$ that $\mathsf{Ev}_i$ is included in $T_B$, verifies the proof against $r_B$, and
verifies the PQ signature on $r_B$. This pattern amortizes the cost of PQ signatures
over many records; the trade-off is the need to manage Merkle proofs and batch
boundaries.

\subsection{Complexity and Storage Analysis}

We now sketch the rough complexity of these patterns in terms of the number of
legacy records $N$, the size of a PQ signature $|\sigma^{\mathrm{pq}}|$, and the batch
size $b$ for Merkle anchoring.

\paragraph{Hybrid signatures.}
For new evidence only, the cost is one additional PQ signature per record. If we
generate $N_{\mathrm{new}}$ evidence items over the remaining lifetime of the system,
the total additional signing cost is $N_{\mathrm{new}}$ PQ signatures; storage overhead is
$N_{\mathrm{new}} \cdot |\sigma^{\mathrm{pq}}|$ bytes.

\paragraph{Re-signing legacy evidence.}
Re-signing all $N$ legacy records individually requires $N$ PQ signatures and adds
$N \cdot |\tau_i|$ bytes of storage, where $|\tau_i|$ is the size of the PQ signature
(and possibly some metadata). The process is embarrassingly parallel and can be
spread over time, but for very large $N$ this may still be expensive.

\paragraph{Merkle-root anchoring.}
If we partition $N$ records into batches of size $b$, we obtain $\lceil N / b \rceil$
roots and therefore need only $\lceil N / b \rceil$ PQ signatures. The storage
overhead for signatures is thus $\lceil N / b \rceil \cdot |\tau_B|$. Each record also
requires storing or reconstructing a Merkle proof of size $O(\log b)$, e.g., a few
dozen hashes for typical batch sizes. The one-time computation cost is dominated by
hashing all records ($N$ hashes) and signing $\lceil N / b \rceil$ roots.

Table~\ref{tab:migration} summarizes the trade-offs qualitatively.

\begin{table}[t]
  \centering
  \begin{tabular}{@{}lccc@{}}
    \toprule
    Pattern & PQ signatures & Per-record overhead & Trust assumptions \\
    \midrule
    Hybrid (new only)      & $N_{\mathrm{new}}$ & $|\sigma^{\mathrm{pq}}|$ & Same as existing system \\
    Re-sign legacy         & $N$               & $|\tau_i|$               & Trusted re-signing process \\
    Merkle-root anchoring  & $\lceil N/b \rceil$ & $O(\log b)$ hashes       & Trusted tree construction \\
    \bottomrule
  \end{tabular}
  \caption{Qualitative comparison of migration patterns for $N$ legacy evidence records.}
  \label{tab:migration}
\end{table}

In practice, combinations of these patterns may be used: for example, Merkle-root
anchoring for very old archival logs and full re-signing for recent, higher-value
records. Recent system designs that combine PQC and blockchain for audit-grade file
transfer~\cite{sola2025qrfile} follow a complementary pattern, anchoring events in an
immutable ledger rather than constant-size evidence records.

\section{Case Study: A Constant-Size Evidence System in Regulated AI}
\label{sec:case}

To illustrate how the above concepts translate into practice, we present a case study
based on a constant-size evidence system implemented in an industrial platform for
regulated AI workloads at Codebat Technologies Inc. The platform targets clinical and
pharmaceutical workflows where each AI job (e.g., a risk prediction or trial cohort
selection) is executed inside a hardened containerized environment and produces a
compact evidence record following the layout in Table~\ref{tab:ev-layout}.

\subsection{Deployment Overview}

In a typical deployment, the platform orchestrates AI jobs across a cluster of
container-based trusted execution environments (TEEs). For each job, the runtime
collects the attestation and provenance attributes described in
Section~\ref{sec:system-context} and produces an evidence record of 320--512~bytes
(depending on signature scheme).

In current deployments, the signature scheme is classical (Ed25519) while the rest
of the design (constant-size layout, hash-based linkage) is agnostic to the choice of
signature. The audit log may grow to $10^6$--$10^8$ evidence records over several
years, depending on workload volume. With a baseline evidence generation throughput
of $2.8 \times 10^5$ records per second on a 16-core server (see Table~\ref{tab:perf}),
the system can handle sustained workloads of thousands of AI jobs per second while
maintaining real-time evidence logging.

\subsection{Applying Migration Patterns}

We now consider how the migration patterns from
Section~\ref{subsec:migration-patterns} apply in this setting.

\paragraph{Hybrid signatures for new evidence.}
Once a PQ signature scheme is available in the platform's cryptographic module, new
jobs can produce hybrid evidence items $(x,\sigma^{\mathrm{class}},\sigma^{\mathrm{pq}})$
without changing the upstream workflow. The primary engineering task is to extend
the evidence generator and verifier to handle an additional signature field. Because
evidence records are already compact (320--512~bytes as in Table~\ref{tab:ev-layout}),
adding a PQ signature (e.g., Dilithium at $\approx 2{,}400$ bytes or SPHINCS+ at
$\approx 7{,}800$ bytes) increases the record size to 3--8~KB, which remains
acceptable for most log storage systems given the strong long-term security
guarantees.

\paragraph{Post-quantum signature performance.}
The choice of PQ signature scheme directly affects evidence record size and
migration throughput. Table~\ref{tab:pq-perf} compares the classical baseline
(Ed25519) with three PQ candidates standardized by NIST~\cite{nistpqc}: ML-DSA
(FIPS~204, lattice-based, formerly CRYSTALS-Dilithium) at two security levels,
and SLH-DSA (FIPS~205, hash-based, formerly SPHINCS+) in its small-signature
variant. Performance figures are single-core reference-implementation estimates
on commodity server hardware (3.0~GHz x86-64), drawn from
recent benchmarks~\cite{raavi2025pqdsa,tan2022pqsignsurvey}.

\begin{table}[t]
  \centering
  \caption{Post-quantum signature scheme comparison for evidence record migration.
           Evidence size = $k \times 32$ bytes (fields) + signature size.
           Performance figures are approximate single-core values from reference
           implementations~\cite{raavi2025pqdsa}.}
  \label{tab:pq-perf}
  \begin{tabular}{@{}lrrrrr@{}}
    \toprule
    Scheme & $|\sigma|$ (B) & $|\mathsf{pk}|$ (B) & Sign/s & Verify/s & Ev size \\
    \midrule
    Ed25519 (baseline)   &     64 &    32 & 35{,}000 & 12{,}000 & 320~B \\
    ML-DSA-65 (Level~3)  &  3{,}309 & 1{,}952 &  5{,}000 & 12{,}000 & 3.6~KB \\
    ML-DSA-87 (Level~5)  &  4{,}627 & 2{,}592 &  3{,}000 &  8{,}000 & 4.9~KB \\
    SLH-DSA-128s (Level~1) & 7{,}856 &    32 &      5 &    200 & 8.1~KB \\
    \bottomrule
  \end{tabular}
\end{table}

ML-DSA-65 offers the best balance for evidence migration: its signing throughput of
$\approx 5{,}000$ operations per second per core is sufficient for re-signing large
corpora, while the $3.3$~KB signature keeps evidence records under $4$~KB. SLH-DSA
provides the most conservative security assumptions (pure hash-based) but is
significantly slower for signing, making it better suited for anchoring small batches
rather than re-signing individual records.

\paragraph{Re-signing recent legacy evidence.}
For recent evidence, especially those corresponding to higher-risk workflows (e.g.,
clinical trial analyses), the platform can re-sign individual records using a PQ
signature key inside a TEE. Using ML-DSA-65 at $\approx 5{,}000$ signatures per
second per core, re-signing $10^7$ records would take on the order of
2{,}000 seconds ($\approx 34$ minutes) on a single core, or under 3 minutes when
parallelized across 16 cores. These estimates suggest that re-signing at this scale
is operationally feasible, even accounting for the $\approx 7\times$ slowdown
compared to classical Ed25519 signatures.

\paragraph{Merkle-root anchoring for archival logs.}
For older archival logs, where individual records are rarely accessed but must remain
provably intact, Merkle-root anchoring can reduce migration cost. For example, with
batch size $b = 2^{12} = 4{,}096$, a corpus of $N = 10^8$ records yields
$\lceil N/b \rceil \approx 24{,}415$ roots. At a PQ signing throughput of $5{,}000$
signatures per second, computing all anchors would take under 5 seconds of CPU
time; the dominant cost lies in computing the $N$ hash values for the leaves, which is
linear in the number of records and can be parallelized. Merkle proofs of depth
$\log_2 b = 12$ hashes add modest per-record overhead.

\subsection{Lessons for Regulated AI Systems}

This case study highlights several lessons:
\begin{itemize}[leftmargin=1.5em]
  \item Designing evidence records to be \emph{cryptographically modular}---i.e., not
        hard-wiring a specific signature scheme into their layout---greatly simplifies
        PQ migration.
  \item For realistic log volumes in regulated AI workloads, the computational cost of
        PQ migration is moderate compared to typical batch processing tasks.
  \item Hybrid and batch anchoring strategies enable operators to prioritize which
        portions of the log receive the strongest PQ protection, aligning cryptographic
        investment with regulatory and business risk.
\end{itemize}

While the concrete numbers above are illustrative rather than definitive, they suggest
that the migration patterns we propose can be implemented without fundamental
scalability obstacles.

\section{Related Work}
\label{sec:related}

\paragraph{Forensics and audit logging.}
A rich body of work studies secure logging, accountability mechanisms, and
digital forensics for regulated environments. Classical secure audit
logs~\cite{securelogging} and forward-secure logging
schemes~\cite{forwardsecurelogging} focus on protecting log integrity against
adversaries who may compromise the logging system after the fact. Foundational
work on cryptographic timestamping~\cite{haberstornetta1991} and authenticated
data structures~\cite{merkle1989certified} provides building blocks for
tamper-evident logs. More recent work explores accountability and audit trails
in cloud and distributed settings, including blockchain-based approaches for
immutable record-keeping~\cite{sola2025qrfile,khan2025pqblockchainaudit}.
Our work complements these approaches by treating constant-size evidence
structures as a first-class cryptographic object, with explicit security
definitions against quantum adversaries and migration patterns tailored to
already-deployed regulated AI systems.

Table~\ref{tab:comparison} provides a systematic comparison of our work with
representative systems from each category.

\begin{table}[t]
  \centering
  \caption{Feature comparison with related work on secure logging and
           post-quantum audit systems. Symbols: \checkmark\ = supported,
           $\circ$ = partial, -- = not addressed.}
  \label{tab:comparison}
  \begin{tabular}{@{}lccccc@{}}
    \toprule
    & \rotatebox{60}{Constant-size} & \rotatebox{60}{PQ signatures}
    & \rotatebox{60}{QROM analysis} & \rotatebox{60}{Formal proofs}
    & \rotatebox{60}{Migration} \\
    \midrule
    CT~\cite{ct}                     & --         & --         & --         & $\circ$    & -- \\
    Somarapu~\cite{somarapu2022qrlogs} & --       & \checkmark & --         & --         & -- \\
    Khan et al.~\cite{khan2025pqblockchainaudit} & -- & \checkmark & --     & --         & -- \\
    Sola-Thomas~\cite{sola2025qrfile} & --         & \checkmark & --         & --         & -- \\
    Yusuf~\cite{yusuf2025qrlogintegrity} & --      & \checkmark & --         & --         & -- \\
    \textbf{This work}               & \checkmark & \checkmark & \checkmark & \checkmark & \checkmark \\
    \bottomrule
  \end{tabular}
\end{table}

\paragraph{Transparency logs.}
Transparency logs such as Certificate Transparency~\cite{ct} and Revocation
Transparency~\cite{transparency} add public append-only structures and gossip
protocols to detect misbehavior by certificate authorities and other trusted
parties. These designs have inspired similar transparency mechanisms in other
domains, including software supply chain and binary transparency. Our evidence
structures share the goal of providing tamper-evident audit trails, but focus
on per-execution attestation in regulated AI settings rather than public
certificate ecosystems.

\paragraph{Post-quantum cryptography.}
Post-quantum cryptography has been extensively studied and is now being
standardized through efforts such as the NIST PQC
project~\cite{bernstein2009pqcbook,nistpqc}. Signature schemes based on
lattices, hash functions, and other post-quantum assumptions are reaching
deployment maturity~\cite{tan2022pqsignsurvey,raavi2025pqdsa}, with growing
support in mainstream cryptographic libraries~\cite{ahmed2025pqclibs}.
Several works have begun to explore quantum-resilient logging and auditing
in domain-specific settings, such as distributed log
storage~\cite{somarapu2022qrlogs}, cloud multimedia
auditing~\cite{khan2025pqblockchainaudit}, telecom logs for
5G/6G~\cite{yusuf2025qrlogintegrity}, and post-quantum TLS
migration~\cite{alnahawi2024pqtls}. Our work provides a unified treatment of
quantum-adversary security notions for constant-size evidence structures
and concrete migration patterns for regulated AI audit trails.

\section{Discussion and Open Problems}
\label{sec:discussion}

Our work raises several broader questions about designing quantum-safe audit trails.

\paragraph{Choice of PQ primitives.}
Different PQ signature schemes offer different trade-offs among key size, signature
size, signing and verification speed, and implementation maturity. Our analysis treats
$\Sigma$ as an abstract EUF-CMA-secure primitive, but deploying such schemes in
evidence systems requires careful engineering and benchmarking. In particular,
regulators may favor conservative, hash-based schemes despite larger signatures, while
performance-sensitive environments may prefer lattice-based schemes.

\paragraph{Policy and retention considerations.}
Cryptographic migration is only one part of a broader policy question: for how long
must evidence be retained, at what integrity level, and under which key rotation
schedule? For example, a regulator may consider it acceptable to rely on classical
signatures for low-risk workflows with short retention periods, while demanding PQ
protection for long-lived, high-impact logs. Formalizing such policies and aligning them
with technical migration patterns is an interesting interdisciplinary challenge.

\paragraph{Beyond integrity: confidentiality and privacy.}
In this paper we focus on integrity, non-repudiation, and binding. Many regulated AI
settings also require protecting the confidentiality of certain log fields and respecting
privacy constraints. Extending evidence structures to include post-quantum encryption
(e.g., via KEM+AEAD constructions) and differential privacy mechanisms, while
maintaining verifiability, is a promising direction.

\paragraph{Stronger quantum security models.}
Our security model accounts for quantum access to hash oracles and EUF-CMA
adversaries against signatures. In more advanced settings, adversaries may have
coherent access to additional system APIs, or evidence structures may be embedded
in protocols that themselves involve quantum communication. Developing stronger
compositional security frameworks for such scenarios remains an open problem.

\paragraph{Richer linkage mechanisms.}
Our Q-Non-Equivocation proof (Theorem~\ref{thm:qne}) covers hash-chain linkage.
Extending this to richer structures such as Merkle-based transparency logs,
skip-list constructions, or DAG-based provenance graphs under quantum adversaries
is an interesting direction. Such extensions would enable quantum-safe
non-equivocation guarantees for a broader class of audit trail architectures
beyond linear chains.

\paragraph{Side channels and implementation security.}
We do not consider side-channel attacks or implementation bugs that leak signing
keys or allow evidence tampering. While these issues are orthogonal to PQ migration,
they are critical in practice. Combining post-quantum cryptography with formally
verified implementations and side-channel-resistant TEEs would further strengthen
audit trails.

\section*{Data and Code Availability}

The prototype implementation used for the microbenchmarks in
Table~\ref{tab:perf} and for the case study in Section~\ref{sec:case}
is part of an industrial platform developed by Codebat Technologies Inc.\
and is currently involved in ongoing patent applications. Due to commercial
and contractual constraints, the source code and deployment configuration
cannot be released at this time. The algorithms, system architecture, and
parameter choices described in Sections~\ref{sec:system-context}
and~\ref{sec:migration} are sufficient to reproduce the reported complexity
trends and to re-implement compatible evidence generators and migration
services on comparable hardware. If and when the relevant patents are granted
and contractual constraints permit, we plan to release a simplified reference
implementation for research and interoperability purposes.

\section{Conclusion}
\label{sec:conclusion}

We have studied evidence structures for regulated AI audit trails in the presence of
quantum adversaries. Building on an abstraction of constant-size evidence records, we
introduced game-based notions of Q-Audit Integrity, Q-Non-Equivocation, and
Q-Binding, and proved via tight reductions that a hash-and-sign instantiation in the
quantum random-oracle model satisfies all three properties under standard assumptions
(collision-resistant hashing and quantum EUF-CMA signatures). We then proposed
migration patterns for deployed evidence logs---hybrid signatures, re-signing of
legacy records, and Merkle-root anchoring---and examined their trade-offs with
concrete post-quantum signature benchmarks.

Our case study, based on an industrial constant-size evidence platform for regulated
AI at Codebat Technologies Inc., suggests that quantum-safe audit trails are
practically achievable with moderate computational and storage overhead. We hope
that our definitions and patterns provide a useful foundation for both researchers and
practitioners designing quantum-resilient compliance systems.

Future work includes extending our analysis to richer linkage mechanisms (e.g.,
Merkle-based transparency logs and DAG-based provenance graphs), integrating
confidentiality and privacy requirements via post-quantum encryption, developing
composable security frameworks for evidence structures embedded in larger protocols,
and collaborating with regulators to align quantum-safe audit designs with evolving
standards.

\bibliographystyle{abbrv}
\bibliography{refs}

\end{document}